\documentclass{llncs}
\usepackage{amsmath, amssymb}
\begin{document}
\pagestyle{headings} \mainmatter
\title{New Quadriphase Sequences
families with Larger Linear Span and Size
}


\author{Wenping Ma }
\institute{ National Key Lab. of ISN, Xidian University , Xi'an 710071, P.R.China\\
\email{wp\_ma@mail.xidian.edu.cn} }

\maketitle

\begin{abstract}
In this paper, new families of quadriphase sequences with larger
linear span and size have been proposed and studied. In particular,
a new family of quadriphase sequences of period $2^n-1$ for a
positive integer $n=em$ with an even positive factor $m$ is
presented, the cross-correlation function among these sequences has
been explicitly calculated. Another new family of quadriphase
sequences of period $2(2^n-1)$ for a positive integer $n=em$ with an
even positive factor $m$ is also presented, a detailed analysis of
the cross-correlation function of proposed sequences has also been
presented.

%\keywords{Quadriphase sequences families \and correlation function
%\and finite field \and Galois ring \and linear span.}
%\subclass{MSC code 94A55 \and MSC code 94A05 }
\end{abstract}

\section{Introduction}
\label{intro} Family of pseudorandom sequences with low cross
correlaton and large linear span has important application in
code-division multiple access communications and cryptology.
Quadriphase sequences are the one most often used in practice
because of their easy implementation in modulators. However, up to
now, only few families of optimal quadriphase sequences are found
\cite{1},\cite{5,6,7,8,9,10,11,12}.

Among the known optimal quadriphase sequence families, the most
famous ones are the families $\mathcal {A}$ and $\mathcal {B}$
investigated by Boztas, Hammons, and Kummar in\cite{5}. The family
$\mathcal {A}$ has period $2^n-1$ and family size $2^n+1$, while the
two corresponding parameters of the family $\mathcal {B}$ are
$2(2^n-1)$ and $2^{n-1}$, respectively. Another optimal family
$\mathcal {C}$ was discussed in \cite{10}, and this family has the
same correlation properties as the family $\mathcal {B}$. Families
$S(t)$ were defined by Kumar et.\cite{6}, and when $t=0$ or $m$ is
odd, the correlation distributions of families $S(t)$ are
established by Kai-Uwe Schmidt\cite{7}. Tang, Udaya, and Fan
generalized the family $\mathcal {A}$ and proposed a new family of
quadriphase sequences with low correlation in \cite{12}. By
utilizing a variation of family $\mathcal {B}$ and $\mathcal {C}$,
Tang and Udaya obtained the family $\mathcal {D}$, which has period
$2(2^n-1)$ and a larger family size $2^n$\cite{9}. Recently, Wenfeng
Jiang, Lei Hu, Xiaohu Tang, and Xiangyong Zeng proposed two new
families $\mathcal {S}$ and $\mathcal {U}$ of quadriphase sequences
with larger linear spans for a positive integer $n=em$ with an odd
positive factor $m$. Both families are asymptotically optimal with
respect to the Wech and Sidelnikov bounds. The family $\mathcal {S}$
has period $2^n-1$, family size $2^n+1$, and maximum correlation
magnitude $2^{\frac{n}{2}}+1$. The family $\mathcal {U}$ has period
$2(2^n-1)$, family size $2^n$, and maximum correlation magnitude
$2^{\frac{n+1}{2}}+2$ \cite{1}.

In this paper, motivated by the constructions proposed in
\cite{1,2,3,4,6}, the new families of quadriphase sequences with
larger linear span and size have been presented. As a special case
of the sequence families, a new family of quadriphase sequences of
period $2^n-1$ for a positive integer $n=em$ with an even positive
factor $m$ is presented, the cross-correlation function among these
sequences has been explicitly calculated. Another new family of
quadriphase sequences of period $2(2^n-1)$ for a positive integer
$n=em$ with an even positive factor $m$ is also presented, a
detailed analysis of the cross-correlation function of proposed
sequences has been presented. The sequences have low correlations
and are useful in code division multiple access  communication
systems and cryptography.

This paper is organized as follows. Section 2 introduces the
preliminaries and notations. In section 3, we give the constructions
and properties of the new sequences families $\mathcal {L}$ and
$\mathcal {V}$ with period $2^{n}-1$. The constructions and
correlation properties of the new sequences family $\mathcal {W}$
with period $2(2^{n}-1)$ are presented in section 4. The conclusions
and acknowledgement are presented in section 5 and 6 respectively.

\section{Preliminaries}
%\label{sec:1} Text with citations \cite{RefB} and \cite{RefJ}.
%\subsection{Subsection title}
%\label{sec:2}
%as required. Don't forget to give each section
%and subsection a unique label (see Sect.~\ref{sec:1}).
\subsection{Basic Concepts}

Let $a=\{a(t)\}$ and $b=\{b(t)\}$  be two quadriphase sequences of
period $L$, the correlation function $R_{a,b}(\tau)$  between them
at a shift $0\leq\tau\leq L-1$ is defined by
$$R_{a,b}(\tau)=\sum_{t=0}^{L-1}\omega^{a(t)-b(t+\tau)}$$
where $\omega^2=-1$.

Let $\mathcal {F}$ be a family of $M$ quadriphase sequences
$$\mathcal {F}=\{a_i=\{a(t)\}:1\leq i\leq M\}.$$

The maximum correlation magnitude $R_{max}$ of $\mathcal {F}$ is
$$R_{max}=\max\{|R_{a_i,a_j}(\tau)|:1\leq i,j\leq M ,\ \ i\neq j\ or\ \tau\neq 0\}.$$

\subsection{Galois Ring}
Let $Z_4[x]$ be the ring of all polynomials over $Z_4$ . A monic
polynomial $f(x)\in Z_4[x]$ is said to basic primitive if its
projection $\overline{f(x)}$
$$\overline{f(x)}=f(x)\ mod\ 2$$
is primitive over $Z_2[x].$

Let $f(x)$ be a basic primitive polynomial of degree $n$ over $Z_4$,
and $Z_4[x]/(f(x))$ denotes the ring of residue classes of
polynomials over $Z_4$ modulo $f(x)$. It can be shown that this
quotient ring is a commutative ring with identity called Galois
ring, denoted as $GR(4,n)$\cite{11}. As a multiplicative group, the
units $GR^*(4,n)$ have the following structure:
$$GR^*(4,n)=G_A\otimes G_C$$
where $G_C$ is a cyclic group of order $2^n-1$  and $G_A$ is an
Abelian group of order $2^n$. Naturally, the projection map
$\overline{a}$ from $Z_4$ to $Z_2$ induces a homomorphism from
$GR(4,n)$ to finite field $GF(2^n)$.

Let $\beta\in GR^*(4,n)$ be a generator of the cyclic group $G_C$ ,
then $\alpha=\overline{\beta}$  is a primitive root of $GF(2^n)$
with primitive polynomial $\overline{f(x)}$  over $Z_2$.

For each element $x\in GR(4,n)$  has a unique $2-adic$
representation of the form
\begin{equation}
\label{eq1}
x=x_0+2x_1,x_0,x_1\in G_C.
\end{equation}

Let $n=em$, The $Frobenius$ automorphism of $GR(4,n)$  over
$GR(4,e)$ is given by
 $$\sigma(x)=x_0^{2^e}+2x_1^{2^e}$$
 for any element $x$ expressed as (\ref{eq1}), and the trace function
 $Tr_e^n$ from $GR(4,n)$ to $GR(4,e)$ is defined by
 $$Tr^n_e(x)=x+\sigma(x)+\sigma^2(x)+\cdots+\sigma^{m-1}(x)$$
 where $\sigma^i(x)=\sigma^{i-1}(\sigma(x))$ for $1<i\leq m-1$.

Let $GF(q)$ is the finite field with $q$ elements, $tr^n_e(x)$ is
the trace function from $GF(2^n)$ to $GF(2^e)$, i.e.,
$$tr_e^n(x)=x+x^{2^e}+\cdots=x^{2^{e(\frac{n}{e}-1)}},x\in GF(2^n).$$
We have $\overline{Tr_e^n(x)}=tr_e^n(\overline{x})$, where $x\in
GR(4,n)$.

Throughout this paper, we suppose (1) $n=em$ with $e\geq 2\ and\
m\geq 2$,(2) $\lambda\in GR(4,e)$ such that $\overline{\lambda}\in
GF(2^e)\setminus\{1,0\}$.
\subsection{Linear Span}
Let
$$f(x)= Tr_1^n[(1+2\alpha)x]+2\sum_{i=1}^rTr_1^{n_i}(A_ix^{v_i}),\alpha\in G_C,A_i\in GF(2^{n_i}),x\in G_C,$$
where $v_i$ is a coset leader of a cyclotomic coset modulo
$2^{n_i}-1$, and $n_i|n$ is the size of the cyclotomic coset
containing $v_i$. For sequence $a=\{a_i\}$ such that
$$a_i=f(\beta),i=0,1,2,\cdots$$where $\beta$ is a primitive
element of $G_C$.

Linear span of a sequence $a$  is equal to $n+\sum_{i,A_i\neq 0}n_i$
, or equivalently, the degree of the shortest linear feedback shift
register that can generates $a$ \cite{1,4}.

\section{ New Quadriphase Sequences with Larger Size and Linear Span }
Define a function $P(x)$ over $GR(4,n)$ as
$$P(x)=\left\{\begin{array}{ll}$$\sum_{j=1}^{l-1}Tr_1^n(x^{2^{ej}+1})+Tr_1^{le}(x^{2^{le}+1}),if\
m=2l$$,\\\\$$\sum_{j=1}^{l}Tr_1^n(x^{2^{ej}+1}),if\
m=2l+1.$$\end{array}\right.$$

For any $x,y\in GR(4,n)$,  it is easy to check that\cite{1,2,3}
\begin{equation}
\label{eq2}
\begin{array}{lll}
2P(x)+2P(y)+2P(x+y) =2Tr_1^n[y(x+Tr^n_e(x))].
\end{array}
\end{equation}

\begin{definition}
Let $\rho$ be an integer such that
$1\leq\rho<\displaystyle{\lfloor\frac{n}{2}\rfloor}$  , a family of
quaternary sequences of period $2^n-1$, $\mathcal {L}=\{s_i(t):0\leq
t<2^{n}-1, 1\leq
i\leq 2^{\rho n}+1\}$ is defined by\\

$s_i(t)=\left\{\begin{array}{ll}$$Tr^{n}_{1}[(1+2\lambda_0^i)\beta^t]+2\sum_{k=1}^{\rho-1}Tr^{n}_{1}(\lambda_k^i\beta^{t(1+2^k)})+2P(\lambda\beta^t),1\leq i\leq 2^{\rho n},$$\\\\
$$2Tr^{n}_{1}(\beta^t),i=2^{\rho n}+1$$
\end{array}\right.$
where
$\{(\lambda_0^i,\lambda_1^i,\cdots,\lambda_{\rho-1}^i),i=1,2,\cdots,2^{\rho
n}\}$ is an enumeration of the elements of $G_C\times G_C\times
\cdots \times G_C$, $\beta$ is a generator element of group $G_C$.
\end{definition}

\begin{lemma}
\label{lem1} All sequence in $\mathcal {L}$  are cyclically
distinct. Thus, the family size of  $\mathcal {L}$  is $2^{n\rho}+1$
.

\end{lemma}
\begin{proof}
The proof of lemma 1 is similar to the proofs of Lemma 1 and Lemma 6
in \cite{4}, we cancel the details.
\end{proof}
\subsection{The Correlation Function of the Sequence Family}
 (1) Suppose $s_i,s_j$,$1\leq i,j\leq 2^{\rho n}$, are two sequences, the correlation function between $s_i$  and $s_j$  is
 \begin{equation}
 \label{eq3}
 R_{s_i,s_j}(\tau)=\sum_{x\in G_{C}}\omega^{Tr^n_1[(1+2\gamma^i_0-(1+2\gamma^j_0)\delta)x]+2\sum_{k=1}^{\rho-1}Tr_1^n(\eta_k x^{1+2^k})+2(P(\lambda x)+P(\lambda\delta x))}-1
\end{equation}
where $\delta=\beta^{\tau}$, $\tau\neq 0$, $\lambda_k^i-\delta
^{1+2^{k}} \lambda^j_k=\eta_k$,$k=1,2,\cdots,\rho-1$.
\begin{displaymath}
\begin{aligned}
&(R_{s_i,s_j}(\tau)+1)(R_{s_i,s_j}(\tau)+1)^* \\
\\
&=\sum_{s\in G_{C}}\sum_{y\in
G_{C}}\omega^{Tr_1^n[(1+2\gamma_{0}^{i}-(1+2\gamma_{0}^{j})\delta)x+2\sum_{k=1}^{\rho-1}Tr_1^n(\eta_kx^{1+2^k})+2(P(\lambda
x)+P(\lambda\delta x))]}\\
&\verb+           +
   \cdot\omega^{-Tr_1^n[(1+2\gamma_{0}^{i}-(1+2\gamma_{0}^{j})\delta)y+2\sum_{k=1}^{\rho-1}Tr_1^n(\eta_ky^{1+2^k})+2(P(\lambda
y)+P(\lambda\delta y))]} \\
\\
&=\sum_{x\in G_{C}}\sum_{y\in
G_{C}}\omega^{Tr_1^n((1+2\gamma_{0}^{i}-(1+2\gamma_{0}^{j})\delta)(x+3y))}\\\\
&\verb+          +
    \cdot\omega^{2[\sum_{k=1}^{\rho-1}Tr_1^n[\eta_k(x^{1+2^k}+y^{1+2^k})]+P(\lambda
x)+P(\lambda\delta x)+P(\lambda y)+P(\lambda\delta y)]} \\
\\
&=\sum_{x\in G_{C}}\sum_{y\in
G_{C}}\omega^{Tr_1^n((\Delta(x+3y))+2(\sum_{k=1}^{\rho-1}Tr_1^n[\eta_k(x^{1+2^k}+y^{1+2^k})]+P(\lambda
x)+P(\lambda\delta x)+P(\lambda y)+P(\lambda\delta y))} \\
\\
&=\sum_{z\in G_{C}}\sum_{y\in G_{C}}\omega^{Tr_1^n(\Delta
z)+2[\sum_{k=1}^{\rho-1}Tr_1^n[\eta_k((y+z+2\sqrt{yz})^{1+2^k}+y^{1+2^k})]+2Tr_1^n(\Delta\sqrt{yz})}\\\\
&\verb+           +
  \cdot\omega^{2P(\lambda(y+z+2\sqrt{yz}))+2P(\lambda\delta(y+z+2\sqrt{yz}))+2P(\lambda
y)+2P(\lambda\delta y)]} \\
\\
&=\sum_{z\in G_{C}}\omega^{\phi(z)}\sum_{y\in
G_{C}}\omega^{2[Tr_1^n(y(\Delta^2z))+v(y,z)]}
\end{aligned}
\end{displaymath}
where
$\Delta=1+2\gamma_{0}^{i}-(1+2\gamma_{0}^{j})\delta$,\ $x=y+z+2\sqrt{yz}$,\\
$\phi(z)=Tr_1^n(\Delta z)+2[P(\lambda z)+2P(\lambda\delta
z)]+2\sum_{k=1}^{\rho-1}Tr^{n}_{1}(\eta_k z^{1+2^k})$,\\\\
$v(y,z)=\sum_{k=1}^{\rho-1}Tr_1^n[\eta_k((y+z+2\sqrt{yz})^{1+2^k}+y^{1+2^k}+z^{1+2^k})]+P(\lambda(y+z+$\\
$2\sqrt{yz}))+P(\lambda\delta(y+z+2\sqrt{yz}))+P(\lambda
y)+P(\lambda\delta z)+P(\lambda z)+P(\lambda\delta z).$\\
Then, by (2), we have
$$2v(y,z)=2Tr_1^n[\lambda y(\lambda z+Tr^n_e(\lambda z))+\lambda\delta y(\lambda\delta z+Tr^n_e(\lambda\delta z))]+2Tr_1^n\sum_{k=1}^{\rho-1}[\eta_k(zy^{2^k}+z^{2^k}y)]$$
$$=2Tr^n_1[y(\lambda^2z+\lambda Tr_e^n(\lambda z)+\lambda^2\delta^2z+\lambda\delta Tr^n_e(\lambda\delta z)+\sum_{k=1}^{\rho-1}(\eta_k^{-2^k}z^{-2^k}+\eta_k z^{2^k}))].$$
Define
\begin{equation}
\label{eq4}
L(z)=\overline{\delta}tr^n_e(\overline{\delta}z)+tr^n_e(z)+\frac{1}{\overline{\lambda}^2}(\overline{\lambda}^2+1)(\overline{\delta}^2+1)z+\frac{1}{\overline{\lambda}^2}\sum_{k=1}^{\rho-1}(\overline{\eta}^{-2^k}_{k}z^{-2^k}+\overline{\eta}_k
z^{2^k}),
\end{equation}
where $z\in GF(2^{n})$, then $L(z)$ is a linear equation over
$GF(2^{n})$.

For $L(z)=0$, we have to count the number of solutions in the
equation
\begin{equation}
\label{eq5}
\overline{\delta}tr^n_e(\overline{\delta}z)+tr^n_e(z)+\frac{1}{\overline{\lambda}^2}(\overline{\lambda}^2+1)(\overline{\delta}^2+1)z+\frac{1}{\overline{\lambda}^2}\sum_{k=1}^{\rho-1}(\overline{\eta}^{-2^k}_{k}z^{-2^k}+\overline{\eta}_k
z^{2^k})=0
\end{equation}
for given $\eta_i$'s in $G_C$,$\lambda\in G_C$ such that
$\overline{\lambda}\in GF(2^e)\backslash\{0,1\}$, and
$\overline{\delta} \in GF(2^{n})\backslash\{0\}$.

It is easy to verify that
$\frac{1}{\overline{\lambda}^2}(\overline{\lambda}^2+1)(\overline{\delta}^2+1)z+\frac{1}{\overline{\lambda}^2}\sum_{k=1}^{\rho-1}(\overline{\eta}^{-2^k}_{k}z^{-2^k}+\overline{\eta}_k
z^{2^k})$ is not a constant polynomial of z, and the maximum number
of solutions of equation $L(z)=0$ is at most $2^{2(\rho-1)+2e}$.
Thus
\begin{equation}
\label{eq6}  |R_{s_i,s_j}(\tau)+1|\leq 2^{\frac{n+2(\rho-1)+2e}{2}}.
\end{equation}
(2) If $1\leq i\leq 2^{\rho n}$ and $j=2^{n\rho}+1$ are two
sequences, then the correlation function between $s_i$ and $s_j$ is
$$R_{s_i,s_j(\tau)}=\sum_{x\in G_C}\omega^{Tr_1^n[(1+2\gamma_1-2\delta)x]+2\sum_{k=1}^{\rho-1}Tr_1^n(\lambda_k^ix^{1+2^k})+2P(\lambda x)}-1,$$
similar to analysis above, the equation (\ref{eq4}) become the
following equation.
$$L(z)=\frac{1+\bar{\lambda}^{2}}{\bar{\lambda}^{2}}z+tr^n_e(z)+\frac{1}{\overline{\lambda}^2}\sum_{k=1}^{\rho-1}((\overline{\lambda}_k^i)^{-2^k}z^{-2^k}+\overline{\lambda}^i_kz^{2^k}).$$
Thus
\begin{equation} \label{eq7}  |R_{s_i,s_j}(\tau)+1|\leq
2^{\frac{n+2(\rho-1)+e}{2}}.
\end{equation}
(3) If $i=j=2^{\rho n}+1$, then $s_i$ is essentially a binary
$m-$sequence, Then $R_{s_i,s_j}(\tau)= -1$ for $\tau\neq 0$.\\\\
(4) Suppose that $s_i$,$s_j$, $1\leq i,j\leq 2^{\rho n}$, are two
sequences, then
$$R_{s_i,s_j}(0)=\sum_{x\in T}\omega^{2Tr_1^n[(\gamma^i_0+\gamma^j_0)x]+2\sum^{\rho-1}_{k=1}Tr_1^n(\eta_kx^{1+2^k})}-1,$$
similar to the analysis above, we have
\begin{equation}
\label{eq8} |R_{s_i,s_j(0)}+1|\leq 2^{\frac{n+2(\rho-1)}{2}}.
\end{equation}
It seems difficult to get tighter bound for
inequality(\ref{eq6})-(\ref{eq8}), thus we propose the following
open problem.\\
 \textbf{Open problem:}For $n=em$ , how many
solutions exist exactly for the equation (\ref{eq5}) over finite
field $GF(2^n)$.

Following the discussion above, we have the following theorem.
\begin{theorem}
For $n=em$ and an integer $\rho$ such that
$1\leq\rho<\displaystyle{\lfloor\frac{n}{2}\rfloor}$, the proposed
quadriphase family has $2^{n\rho}+1$ cyclically distinct binary
sequences of period $2^n-1$. The maximum correlation magnitude of
sequences is smaller than $1+2^{\frac{n+2(\rho-1)+2e}{2}}$.
Therefore, the sequences family constitutes a
$(2^n-1,2^{n\rho}+1,1+2^{\frac{n+2(\rho-1)+2e}{2}})$ quadriphase
signal set.
\end{theorem}
\subsection{Linear Spans of the Sequence}
In order to express clearly, let

$s(\lambda_{0},\Lambda,t)=Tr^{n}_{1}[(1+2\lambda_0)\beta^t]+2\sum_{k=1}^{\rho-1}Tr^{n}_{1}(\lambda_k\beta^{t(1+2^k)})+2P(\lambda\beta^t)
$,\\ where $\Lambda=(\lambda_1,\cdots,\lambda_{\rho-1})$.

We divide the set $\Delta=\{1,2,\cdots,\rho-1\}$ into two sets $A$
and $B$ such that $\Delta=A\bigcup B$, where $A=\{ke+r:1\leq
k\leq\displaystyle{ \lfloor\frac{\rho-1}{e}\rfloor,0<r<e}\}$,
$B=\{ke:1\leq k\leq\displaystyle{\lfloor\frac{\rho-1}{e}\rfloor}\}$.
\begin{theorem}
(1)Consider a sequence represented by $s(\lambda_{0}, \Lambda,t)$
where $j\ \lambda_i$'s with $i\in A$ in
$\Lambda=(\lambda_1,\cdots,\lambda_{\rho-1})$ are equal to 0 and $l\
\lambda_i$'s with $i\in B$ in
$\Lambda=(\lambda_1,\cdots,\lambda_{\rho-1})$ are equal to
$\overline{\lambda}$.Let $LS_{j,l}(\rho)$ be the linear span of the
sequence.Then
$$LS_{j,l}(\rho)=n(\frac{m-1}{2}+\rho-1-\lfloor\frac{\rho-1}{e}\rfloor+1-j-l),0\leq j\leq |A|,0\leq l\leq |B|.$$
and there are $(^{\rho-1-\lfloor\frac{\rho-1}{e}\rfloor}_{\verb+   +
j})(_{\verb+  +
 l}^{\lfloor\frac{\rho-1}{e}\rfloor})2^{n} (2^n-1)^{\rho-1-j-l}$
sequences having linear span $LS_{j,l}(\rho)$. \\\\
(2) The linear span of the sequences $s_{2^{n\rho}+1}(t)$ is $n$.
\end{theorem}
\begin{proof}
First, consider the linear span of sequences with $m$ is odd. A
sequence constructed above has a total of
$\displaystyle{\frac{m-1}{2}+\rho-1-\lfloor\frac{\rho-1}{e}\rfloor+1}$
trace terms and each trace term has the linear span of $n$. If $j\
\lambda$'s with $i\in A$ in
$\Lambda=(\lambda_1,\cdots,\lambda_{\rho-1})$ are equal to 0, and
$l\ \lambda_i$'s with $i\in B$ in
$\Lambda=(\lambda_1,\cdots,\lambda_{\rho-1})$ are equal to
$\overline{\lambda}^ {2^{i}+1}$, it has
$\displaystyle{\frac{m-1}{2}+\rho-1-\lfloor\frac{\rho-1}{e}\rfloor+1-j-l}$
nonzero trace terms and the corresponding linear span of the
sequences is given by
$$LS_{j,l}(\rho)=n(\frac{m-1}{2}+\rho-1-\lfloor\frac{\rho-1}{e}\rfloor+1-j-l),0\leq j\leq |A|,0\leq l\leq |B|.$$
Since (1) $j\ \lambda_i$'s with $i\in A$  are 0 and $(|A|-j)\
\lambda$ 's are nonzero, (2) $l\ \lambda_{i}$'s with $i\in B$ are
$\overline{\lambda}^{2^{i}+1}$ and $(|B|-l)\ \lambda_{i}$'s are not
equal to $\overline{\lambda}^{2^{i}+1}$, (3) the number of
$\lambda_{0}$ is $2^{n}$. Therefore, the number of corresponding
sequences given above is
\begin{displaymath}
\begin{aligned}
&(_{\verb+   +
j}^{\rho-1-\lfloor\frac{\rho-1}{e}\rfloor})(2^n-1)^{\rho-1-\lfloor\frac{\rho-1}{e}\rfloor-j}\cdot
(_{\verb+  +
 l}^{\lfloor\frac{\rho-1}{e}\rfloor})(2^n-1)^{\lfloor\frac{\rho-1}{e}\rfloor-l} \cdot 2^{n}
\\
\\
&=(_{\verb+  +j}^{\rho-1-\lfloor\frac{\rho-1}{e}\rfloor})\cdot
(_{\verb+  + l}^{\lfloor\frac{\rho-1}{e}\rfloor})2^{n}
(2^n-1)^{\rho-1-l-j}.
\end{aligned}
\end{displaymath}
Applying this result to each $j$ and each $l$, we obtain the linear
span of the proposed sequence. Using a similar approach to the odd
case, we see that the linear span of both sequences is same.
\end{proof}

For the sequences families above, some special conditions had
already been discussed, for example, the case with $n=em$, where $m$
is an odd, and $\rho=1$ had been discussed in \cite{1}. In the
following, we will discuss another special case with $n=em$, where
$m$ is an even, and $\rho=1$, we call this special sequence family
as family $\mathcal {V}$.
\subsection{Correlation Function of the Sequence family for even $m$ and $\rho=1$}
If $\rho=1$, then the equation (\ref{eq5}) becames
\begin{equation}
\label{eq9}
\overline{\delta}Tr^n_e(\overline{\delta}z)+Tr^n_e(z)+\frac{(\overline{\lambda}^2+1)(\overline{\delta}^2+1)}{\overline{\lambda}^2}z=0
\end{equation}
In the following, we study the solution of the equation (\ref{eq9}).
\\\\
Let $Tr_e^n(\overline{\delta}z)=a$, $Tr_e^n(z)=b$, then
$$z=\frac{\overline{\lambda}^2}{\overline{\lambda}^2+1}\frac{\overline{\delta} a +b}{\overline{\delta}^2+1}.$$
By computing $Tr^n_e(z)$ and $Tr^n_e(\overline{\delta}z)$, we
have\\
\begin{equation}
\label{eq10}
\left\{\begin{array}{lll}$$aTr^n_e(\displaystyle{\frac{\overline{\delta}}{\overline{\delta}^2+1}})+b[Tr^n_e(\frac{1}{\overline{\delta}^2+1})-\frac{\overline{\lambda}^2+1}{\overline{\lambda}^2}]=0$$\\\\
$$\displaystyle{a[tr^n_e(\frac{1}{\overline{\delta}^2+1})-\frac{\overline{\lambda}^2+1}{\overline{\lambda}^2}]+bTr^n_e(\frac{\overline{\delta}}{\overline{\delta}^2+1})=0}$$
\end{array}\right.
\end{equation}
The determinant of corresponding coefficient matrix of (\ref{eq10})
is equal to
\begin{displaymath}
\begin{aligned}
&[Tr^n_e(\frac{1}{\overline{\delta}^2+1})-\frac{\overline{\lambda}^2+1}{\overline{\lambda}^2}]^2+[Tr^n_e(\frac{\overline{\delta}}{1+\overline{\delta}^2})]^2\\\\
&=(\frac{\overline{\lambda}^2+1}{\overline{\lambda}^2})^2+[Tr^n_e(\frac{1}{1+\overline{\delta}})]^2
\end{aligned}
\end{displaymath}
(1) If $\displaystyle{Tr^n_e(\frac{1}{1+\overline{\delta}})\neq
\frac{\overline{\lambda}^2+1}{\overline{\lambda}^2}}$, then  the
determinant of coefficient matrix (\ref{eq10}) is not equal to zero,
the equation (\ref{eq10}) has unique solution $a=0$, $b=0$, then
$z=0$.\\
(2) If $\displaystyle{Tr^n_e(\frac{1}{1+\overline{\delta}})=
\frac{\overline{\lambda}^2+1}{\overline{\lambda}^2}}$, then the
determinant of coefficient matrix (\ref{eq10}) is equal to zero,
\begin{displaymath}
\begin{aligned}
&aTr^n_e(\frac{1}{\overline{\delta}+1})+aTr^n_e(\frac{1}{\overline{\delta}^2+1})+b[Tr^n_e(\frac{1}{\overline{\delta}^2+1})-\frac{\overline{\lambda}^2+1}{\overline{\lambda}^2}]=0,\\
\end{aligned}
\end{displaymath}
thus $a=b$, the equation (\ref{eq10}) has $2^e$ solutions, then
$z=\frac{\overline{\lambda}^2}{\overline{\lambda}^2+1}\frac{1}{\overline{\delta}+1}a$.

\begin{displaymath}
\begin{aligned}
&2P(\overline{\lambda}z)+2P(\overline{\lambda}\overline{\delta}z)=2\sum^{l-1}_{j=1}Tr^n_1[(\overline{\lambda}z)^{2^{ej}+1}+(\overline{\lambda}\overline{\delta}z)^{2^{ej}+1}]+2Tr^{le}_1[(\overline{\lambda}x)^{2^{le}+1}+(\overline{\lambda}\overline{\delta}z)^{2^{el}+1}]\\
&=2\sum^{l-1}_{j=1}Tr^n_1[(\frac{\overline{\lambda}^3a}{1+\overline{\lambda}^2})^{2^{ej}+1}(\frac{1}{1+\overline{\delta}})^{2^{ej}+1}+(\frac{\overline{\lambda}^3a}{1+\overline{\lambda}^2})^{2^{ej}+1}(\frac{\overline{\delta}}{1+\overline{\delta}})^{2^{ej}+1}]\\
&\verb+  ++2Tr^{el}_1[(\frac{\overline{\lambda}^3a}{1+\overline{\lambda}^2})^{2^{el}+1}(\frac{1}{1+\overline{\delta}})^{2^{el}+1}+(\frac{\overline{\lambda}^3a}{1+\overline{\lambda}^2})^{2^{el}+1}(\frac{\overline{\delta}}{1+\overline{\delta}})^{2^{el}+1}]\\
&=2Tr^e_1\{(\frac{\overline{\lambda}^3a}{1+\overline{\lambda}^2})^2[\sum^{l-1}_{j=1}Tr^n_e[(\frac{1}{1+\overline{\delta}})^{2^{ej}+1}+(\frac{\overline{\delta}}{1+\overline{\delta}})^{2^{ej}+1}]\\
&\verb+       + +Tr^{le}_e[(\frac{1}{1+\overline{\delta}})^{2^{el}+1}+(\frac{\overline{\delta}}{1+\overline{\delta}})^{2^{el}+1}]]\}\\
&=2Tr^{e}_{1}{[(\frac{\overline{\lambda}^{3}a}{1+\overline{\lambda}^{2}})^{2}}(Tr^{le}_{e}1+Tr^{n}_{e}(\frac{1}{1+\overline{\delta}}))]\\
&=\left\{\begin{array}{ll}$$2Tr^{n}_{1}(\frac{\overline{\lambda}^{2}z}{1+\overline{\lambda}^{2}}),\
for \ odd \ l
$$,\\\\$$2Tr^{n}_{1}(\frac{\overline{\lambda}^{2}z}{1+\overline{\lambda}}), \ for
\ even  \ l .$$\end{array}\right. \\
\end{aligned}
\end{displaymath}

Thus, $\phi(z)=Tr^n_1(\Delta z)+2[P(\lambda z)+P(\lambda\delta z)]$
$$=\left\{\begin{array}{ll}$$2Tr^{n}_{1}[(\overline{\gamma}_{0}^{i}+(\overline{\gamma}_{0}^{j}+1)\overline{\delta}+\frac{\overline{\lambda}^{2}}{1+\overline{\lambda}^{2}})z],\
\ for \ odd \ l $$ ,\\\\$$
2Tr^{n}_{1}[(\overline{\gamma}_{0}^{i}+(\overline{\gamma}_{0}^{j}+1)\overline{\delta}+\frac{\overline{\lambda}^{2}}{1+\overline{\lambda}})z],
\ for \ even \ l.$$ \end{array}\right. $$

Because the solutions space of equation (9) is a linear subspace,
following the discussions above, we have
\begin{theorem}
for $m$ is even, $\rho=1$, the nontrivial correlation function of
the proposed sequences family $\mathcal {V}$ takes values in
$\{-1,-1\pm 2^{\frac{n}{2}},-1\pm 2^{\frac{n}{2}}\omega,-1\pm
2^{\frac{n+e}{2}},-1\pm 2^{\frac{n+e}{2}}\omega \}$.
\end{theorem}
\section{Quadriphase Sequences with period $2(2^n-1)$ }
Similar to the \cite{1,2}, for an even $m$, we propose the following
sequence family, the correlation function of the sequences family is
calculated.

In this section, let $G=\{\eta_1,\eta_2,\cdots,\eta_{2^{n-1}}\}$ be
a maximum subset of $G_C$ such that $2\eta_i\neq 2(\eta_j+1)$ for
arbitrary $1\leq i,j\leq 2^{n-1}$. By convention, denote
$\beta^{\frac{1}{2}}=\beta^{2^{n-1}}$. We present another family of
quadriphase sequences with period $2(2^n-1)$ as follows.
\begin{definition}
A family $\mathcal {W}$ of quadriphase sequences with period
$2(2^n-1)$ is
defined as $\mathcal {W}=\{u_i(t),v_i(t):0\leq i<2^{n-1}\}$ is given by\\
 1)\\
%\begin{displaymath}
$u_i(t)=$$\left\{ \begin{array}{lll}
Tr^n_1[(1+2\eta_i)\beta^{t_0}]+2P(\lambda\beta^{t_0}),t=2t_0\\\\
Tr^n_1[(1+2(\eta_i+1))\beta^{t_0+\frac{1}{2}}]+2P(\lambda\beta^{t_0+\frac{1}{2}}),t=2t_0+1
\end{array}\right.$\\
%\end{displaymath}
for $0\leq i<2^{n-1}$, where $\eta_i\in G$.\\
2)\\
%\begin{displaymath}
$v_i(t)=$$\left\{ \begin{array}{lll}
Tr^n_1[(1+2\eta_i)\beta^{t_0}]+2P(\lambda\beta^{t_0})+2,t=2t_0\\\\
Tr^n_1[(1+2(\eta_i+1))\beta^{t_0+\frac{1}{2}}]+2P(\lambda\beta^{t_0+\frac{1}{2}}),t=2t_0+1
\end{array}\right.$\\
%\end{displaymath}
for $0\leq i<2^{n-1}$, where $\eta_i\in G$.
\end{definition}
\begin{theorem}
the correlation functions of the family $\mathcal {W}$ satisfy the
following properties. 1) if $\tau=2^n-1$, then
$R_{u_i,u_j}(\tau)=-2$,
$R_{v_i,v_j}(\tau)=2$, $R_{u_i,v_j}(\tau)=0$. \\
2) if $\tau=0$, then $R_{u_i,v_j}(\tau)=0$ and
\begin{displaymath}
R_{u_i,u_j}(\tau)=R_{v_i,v_j}(\tau)=\left\{
\begin{array}{ll}
2(2^n-1),i=j\\
-2,i\neq j,
\end{array}\right.
\end{displaymath}
3)If $\tau=2\tau_0+1\neq 2^n-1$, then
\begin{displaymath}
\begin{aligned}
&a)\ R_{u_i,u_j}(\tau)\ takes\ values\ in\ \{-2,-2\pm
2^{\frac{n}{2}+1},-2\pm 2^{\frac{n+e}{2}+1}\},\\
&b)\ R_{v_i,v_j}(\tau)\ takes\ values\ in\ \{2,2\pm
2^{\frac{n}{2}+1},2\pm 2^{\frac{n+e}{2}+1}\},\\
&c)\ R_{u_i,v_j}(\tau)\ takes\ values\ in\ \{\pm
2^{\frac{n}{2}+1}\omega,\pm2^{\frac{n+e}{2}+1}\omega\}.
\end{aligned}
\end{displaymath}
4)If $\tau=2\tau_0$ and $\tau_0\neq 0$, then
\begin{displaymath}
\begin{aligned}
&a)\ R_{u_i,u_j}(\tau)\ takes\ values\ in\ \{-2,-2\pm
2^{\frac{n}{2}+1},-2\pm 2^{\frac{n+e}{2}+1}\},\\
&b)\ R_{v_i,v_j}(\tau)\ takes\ values\ in\ \{-2,-2\pm
2^{\frac{n}{2}+1},-2\pm 2^{\frac{n+e}{2}+1}\},\\
&c)\ R_{u_i,v_j}(\tau)\ takes\ values\ in\ \{\pm
2^{\frac{n}{2}+1}\omega,\pm2^{\frac{n+e}{2}+1}\omega\}.
\end{aligned}
\end{displaymath}
\end{theorem}
\begin{proof}
In order to analysis easily, let
\begin{equation}
\label{eq11} \varsigma(\gamma_1,\gamma_2,\delta)=\sum_{x\in
G_{C}}\omega^{Tr^n_1[(1+2\gamma_1-(1+2\gamma_2)\delta)x]+2(P(\lambda
x)+P(\lambda\delta x))}.
\end{equation}
It is easy to check that
$\varsigma(\gamma_1+1,\gamma_2,\delta)=\varsigma(\gamma_1,\gamma_2+1,\delta)^*$, where $*$ denotes complex conjugate. \\
Similar to \cite{1}, the following facts can be easily checked. \\
1) if $\tau=2\tau_0+1$, then
\begin{displaymath}
\begin{aligned}
&R_{u_i,u_j}(\tau)=\varsigma(\eta_i,\eta_j+1,\delta)+\varsigma(\eta_i+1,\eta_j,\delta)-2,\\
&R_{v_i,v_j}(\tau)=-\varsigma(\eta_i,\eta_j+1,\delta)-\varsigma(\eta_i+1,\eta_j,\delta)+2,\\
&R_{u_i,v_j}(\tau)=\varsigma(\eta_i,\eta_j+1,\delta)-\varsigma(\eta_i+1,\eta_j,\delta).
\end{aligned}
\end{displaymath}
2) if $\tau=2\tau_0$, then
\begin{displaymath}
\begin{aligned}
&R_{u_i,u_j}(\tau)=\varsigma(\eta_i,\eta_j+1,\delta)+\varsigma(\eta_i+1,\eta_j,\delta)-2,\\
&R_{v_i,v_j}(\tau)=\varsigma(\eta_i,\eta_j+1,\delta)+\varsigma(\eta_i+1,\eta_j,\delta)-2,\\
&R_{u_i,v_j}(\tau)=-\varsigma(\eta_i,\eta_j+1,\delta)+\varsigma(\eta_i+1,\eta_j,\delta).
\end{aligned}
\end{displaymath}
Due to (\ref{eq3}),(\ref{eq11}) and the theorem 3, the theorem 4 is
proved.
\end{proof}
Similar to the proof of the theorem 2 above, or the proof of theorem
3 and theorem 7 \cite{1} the following theorem is obtained.
\begin{theorem}
the linear spans of the sequences in $\mathcal {W}$ are given as follows\\
(1) For $u_i\in \mathcal {W}$, the linear span $LS(u_i)$ of $u_i$ is
given by
$\displaystyle{LS(u_i)=\frac{n(n+e)}{2e}}$.\\
(2) For $v_i\in \mathcal {W}$, the linear span $LS(v_i)$ of $v_i$ is
given by $\displaystyle{LS(u_i)=\frac{n(n+e)}{2e}+2}$.
\end{theorem}
\section{Conclusions}
In this paper, we have proposed the new families of quadriphase
sequences with larger linear span and size. The maximum correlation
magnitude of proposed sequences family is bigger then that of the
related sequence in \cite{1}, and is smaller than that of the
related binary sequences family in \cite{2,3} with same parameters.
The proposed two families of quadriphase sequences with period
$2^n-1$ and $2(2^n-1)$ respectively for a positive integer $n=em$
where $m$ is an even positive can be take as an extensions of the
results in \cite{1} where $m$ is an odd positive.

\section {Acknowledgment}
This work was supported by National Science Foundation of China
under grant No.60773002 and 61072140, the Project sponsored by $SRF$
for $ROCS$, $SEM$, 863 Program (2007AA01Z472), and the 111 Project
(B08038).
%\paragraph{Paragraph headings} Use paragraph headings as needed.

%\begin{figure}
%  \includegraphics{example.eps}
%\end{figure}
%\begin{figure*}
%  \includegraphics[width=0.75\textwidth]{example.eps}
%\end{figure*}
%\begin{table}
%\begin{tabular}{lll}
%\hline\noalign{\smallskip}
%\end{tabular}
%\end{table}

%\begin{acknowledgements}
%\end{acknowledgements}



\end{document}